\documentclass[10pt,aps,pra,superscriptaddress,showpacs]{revtex4}

\usepackage{url}
\usepackage{amsmath}
\usepackage{amssymb}
\usepackage{amsthm}
\usepackage{amsfonts}
\usepackage{braket}
\usepackage{complexity}
\usepackage{graphicx}% Include figure files
\usepackage{dcolumn}% Align table columns on decimal point
\usepackage{bm}% bold math
\usepackage{hyperref}
\usepackage{enumerate}
\usepackage{algorithm}
\usepackage{color}
\usepackage{algpseudocode}

\usepackage{multirow}

\input{Qcircuit}

\newtheorem{theorem}{Theorem}

\newcommand{\eq}[1]{\hyperref[eq:#1]{(\ref*{eq:#1})}}
\renewcommand{\sec}[1]{\hyperref[sec:#1]{Section~\ref*{sec:#1}}}
\newcommand{\app}[1]{\hyperref[app:#1]{Appendix~\ref*{app:#1}}}
\newcommand{\fig}[1]{\hyperref[fig:#1]{Figure~\ref*{fig:#1}}}
\newcommand{\thm}[1]{\hyperref[thm:#1]{Theorem~\ref*{thm:#1}}}
\newcommand{\lem}[1]{\hyperref[lem:#1]{Lemma~\ref*{lem:#1}}}
\newcommand{\tab}[1]{\hyperref[tab:#1]{Table~\ref*{tab:#1}}}
\newcommand{\cor}[1]{\hyperref[cor:#1]{Corollary~\ref*{cor:#1}}}
\newcommand{\alg}[1]{\hyperref[alg:#1]{Algorithm~\ref*{alg:#1}}}
\newcommand{\defn}[1]{\hyperref[def:#1]{Definition~\ref*{def:#1}}}

\newcommand{\tout}[1]{{}}

\begin{document}

\title{Randomized gap and amplitude estimation}
\author{Ilia Zintchenko}
\affiliation{Quantum Architectures and Computation Group, Microsoft
  Research, Redmond, WA 98052, USA}
\author{Nathan Wiebe}
\affiliation{Quantum Architectures and Computation Group, Microsoft
  Research, Redmond, WA 98052, USA}

\begin{abstract}
We provide a new method for estimating spectral gaps in
low-dimensional systems. Unlike traditional phase estimation, our
approach does not require ancillary qubits nor does it require well
characterised gates. Instead, it only requires the ability to perform
approximate Haar--random unitary operations, applying the unitary
whose eigenspectrum is sought out and performing measurements in the
computational basis. We discuss application of these ideas to in-place
amplitude estimation and quantum device calibration.
\end{abstract}  

\maketitle

\section{Introduction}
In recent years a host of methods have been developed for performing phase
estimation in quantum systems~\cite{pe0,pe1,pe2,bonato2015optimized}.  These methods, driven by demand in
quantum algorithms and metrology, have provided ever more efficient means
of learning the eigenvalues of a quantum system.  Through the use of sophisticated ideas
from optimization and machine learning, recent phase estimation methods come
close to saturating the Heisenberg bound~\cite{higgins2007entanglement,hentschel2010machine,bonato2015optimized,pe2}, which is the ultimate performance limit
for any phase estimation algorithm.  

Despite the near-optimality of existing methods, there are a number of avenues of inquiry that still remain
open surrounding phase estimation.  Importantly, recent work has begun to look at operational restrictions
on phase estimation including limited visibility in measurements, de-coherence and time-dependent drift.
Such generalisations are especially important as we push towards building a scalable quantum computer
and are faced with the challenge of characterising the gates in a massive quantum system that is 
imperfectly calibrated.

One important restriction that has not received as much attention is the issue that traditional phase
estimation algorithms require entangling the quantum device with one (or more) ancilla qubits.
This need to couple the system with an external qubit precludes ideas from phase estimation
to be applied to single qubit devices.  Even if ancilla qubits are available, such characterisation
methods would entail the use of entangling gates which are very costly in many quantum computing
platforms.

Our work provides a way to circumvent this problem for small quantum
devices.  It allows an experimentalist to learn the spectral gaps in
an uncharacterised device with an amount of experimental time that is
proportional to that of phase estimation without requiring any
ancillary qubits.  The idea of our approach is reminiscent of
randomised
benchmarking~\cite{DCE09,kimmel2014robust,kimmel2015robust,magesan2011scalable},
and particularly iterated interleaved benchmarking~\cite{sheldon2016iterative}, in that we use random operations to extract
information about the underlying dynamics of the system.

We consider three
applications for this method, which we call~\emph{randomised gap estimation}:
\begin{enumerate}
\item Eigenspectrum estimation for small quantum systems that lack
  well characterised gates.
\item Ancilla--free amplitude estimation.
\item Control map learning for small quantum devices.
\end{enumerate}
This last application is interesting in the context of quantum
bootstrapping~\cite{qp1} because it gives an inexpensive way of
calibrating a set of one and two qubit gates at the beginning of a
bootstrapping protocol. This is significant because an important
caveat in bootstrapping is that while a trusted simulator can be used
to quickly learn how to control an uncalibrated system, calibrating
that trusted simulator requires exponentially more measurements and
polynomially more experimental time than using traditional methods.
Thus our techniques can, under some circumstances, dramatically reduce
the total cost of a bootstrapping protocol.

Our paper is laid out as follows.  We first provide a review of
Bayesian inference in~\sec{bayes} and discuss the method for
approximate Bayesian inference that we use in our numerical studies.
We then introduce randomised gap estimation in~\sec{gap} and discuss
the challenges faced when applying the method to high--dimensional
systems.  We then show how the adiabatic theorem can be used to
eliminate this curse of dimensionality for certain systems
in~\sec{adiabatic}.  We then apply randomised gap estimation to the
problem of amplitude estimation in~\sec{amplitude} and also to
learning a map between experimental controls and the system
Hamiltonian for a quantum device in~\sec{map} before concluding.

\section{Bayesian inference}
\label{sec:bayes}
Bayesian inference is a widely used method to extract information from
a system. The goal of Bayesian inference is to compute the probability
that a hypothesis $x$ is true given evidence $E$ and a set of prior
beliefs about the hypotheses. These prior beliefs are represented as a
probability distribution known as a prior. This prior, along with
evidence $E$, can be thought of as an input to the inference
algorithm. The output of the algorithm is the posterior distribution
which is given by Bayes' theorem as
\begin{equation}
P(x|E) = \frac{P(E|x) P(x)}{\int P(E|x) P(x) \mathrm{d}x},\label{eq:bayes}
\end{equation}
where $P(E|x)$ is known as the likelihood function. The posterior
distribution output by this process is then used as a prior
distribution in online inference algorithms and this process of
generating a new prior distribution using Bayes' rule is known as a
Bayesian update.

As an example of how Bayesian inference plays a role in eigenvalue estimation, 
consider the traditional approach used for iterative phase estimation.
In iterative phase estimation, one wishes to learn the eigenvalues
of a unitary operation $U$ eigenstates using the following circuit:
\[
    \Qcircuit @C=1em @R=1em {
        \lstick{\ket{0}}    & \gate{H}  & \gate{e^{-iM \theta Z}}   & \ctrl{1}   & \gate{H} & \meter & \cw&E \\
        \lstick{\ket{\phi}} & {/} \qw   & \qw                   & \gate{U^M} & \qw      & \qw    & \qw&
    }
\]
where $M$ and $\theta$ are user--specifiable experimental parameters, $H$ is the Hadamard matrix and $Z$ is the Pauli--$Z$ operator.  
Let $\ket{\phi}$ be an eigenstate such that $U\ket{\phi} = e^{i\lambda} \ket{\phi}$.  The likelihood function for this circuit is then
\begin{equation}
P(E=0|\lambda;\theta,M) = \cos^2(M(\theta -\lambda)).
\end{equation}

Note that in the above circuit we assume that $\ket{\phi}$ is an eigenstate.  If it is not then the protocol will allow us to learn an eigenvalue corresponding to an eigenvector within the support of $\ket{\phi}$~\cite{pe0}.

The problem of phase estimation is thus the problem of estimating $\lambda$ given a sequence of different experimental outcomes.
This can be done using~\eq{bayes} starting from an initial prior distribution for $\lambda$ that is uniform over $[0,2\pi)$.
A further benefit is that the posterior variance provides an estimate of the uncertainty in the inferred phase as well as an 
estimate of the most likely value of $\lambda$.
This approach is considered in~\cite{pe1,pe2} wherein exact Bayesian inference is found to perform extremely well
both in settings where $\theta$ and $M$ are chosen adaptively as well as non--adaptively.

If $P(x)$ has support over an infinite number of points then exact
Bayesian inference is usually intractable and discretisations are
often employed to address this problem. Such discretisations include
particle filter methods and sequential Monte Carlo
methods~\cite{BK98,DGA00,LW01}.  Here we employ a newly
developed approach known as rejection filter inference~\cite{pe2,WGKS15}.

 The idea behind Rejection filtering is to use rejection
sampling to convert an ensemble of samples from the prior distribution
to a smaller set of samples from the posterior distribution.
Specifically, if evidence $E$ is observed and we draw a sample from
the prior distribution and accept it with probability equal to
$P(E|x)$ then the probability that hypothesis $x$ is accepted as a
sample is from Bayes' theorem
\begin{equation}
P(E|x) P(x) \propto P(x|E)
\end{equation}
Therefore, the samples that pass through the rejection filter are
distributed according to the posterior distribution.

Although this process allows us to sample from the posterior
distribution, it is not efficient. This is because every time we
perform an update, we will, on average, reject a constant fraction of
samples. This means that the number of samples kept will shrink
exponentially with the number of updates. We can, however, make this
process efficient by fitting these samples to a family of
distributions and then draw a new set of samples from this model
distribution for the next update. This ability to regenerate samples
allows rejection filtering inference to avoid this catastrophic loss
of support for the approximation to the posterior.

There are a number of models for the prior and posterior that can be
considered. Here we use a unimodal Gaussian distribution. 
Gaussian models for the posterior distribution provide a number of
advantages. First, they are parametrised by the posterior mean and
covariance matrix which give an estimate of the true hypothesis and
the uncertainty in it. Furthermore, these quantities are easy to
estimate from the accepted samples and can be computed incrementally,
which allows our algorithm to be executed using near-constant memory.

There are several advantages to using rejection filtering for approximate inference.
Specifically, it is very fast,
easy to parallelise, can be implemented using far less memory than
particle filter or sequential Monte-Carlo methods.  Perhaps most importantly, it is also 
substantially easier to implement than traditional particle filters. Rejection filtering has also
been successfully used in phase estimation algorithms~\cite{pe2}, where the use
of rejection filtering leads to substantial reductions in the
experimental time needed to perform the inference relative to Kitaev's
phase estimation algorithm~\cite{pe0} and information theoretic phase
estimation~\cite{pe1}. 

\section{Randomized gap estimation}
\label{sec:gap}

\begin{algorithm}[t!]
  \caption{Bayesian inference algorithm for eigenphases}
  \label{alg:inference0}
  \begin{algorithmic}
    \Require Set of unitary operators $\{U: j=1:K\}$, evolution time
    $t$, prior over gaps $P(\Delta)$.

    \State Prepare state $\ket{0}\in \mathbb{C}^N$.

    \State Randomly select $U$ from set of unitary operations.

    \State $\ket{0}\gets U^\dagger e^{-iHt} U \ket{0}$.

    \State Measure state and $E\gets 0$ if the result is '0' otherwise
    $E\gets 1$.

    \State Use approximate Bayesian inference to estimate
    $P(\Delta|E)\propto P(\Delta) \cdot \int P(E|\Delta;t,U)
    \mu(U)\mathrm{d}U$.

    \State\Return $P(\Delta|E)$.
  \end{algorithmic}
\end{algorithm}

Traditional approaches for learning the eigenspectrum of a Hamiltonian
require ancillary qubits and well characterised gates. Here we present
an approach, which we call randomised gap estimation, to efficiently
estimate the eigenspectrum of a small system with no ancillary qubits
and potentially poorly characterised gates. The idea behind this
approach is to use Bayesian inference in concert with random
evolutions to infer the gaps between the eigenphases of a unitary
operator. 

In the first step, the state
\begin{equation}
|\Psi\rangle:=U \ket{0} = \sum_k \beta_k |v_k\rangle
\end{equation}
is prepared, where $\ket{v_k}$ is the eigenvector corresponding to
eigenvalue $\lambda_k$ of $H$ in an arbitrary ordering. Here $\beta_k$
are unknown parameters that depend not only on the random unitary
chosen, but also the eigenbasis of $H$. For low--dimensional systems
exact methods for drawing $U$ uniformly according to the Haar measure
are known (more generally it suffices to draw $U$ from a unitary
$2$--design see~\sec{haar_random} for more details). For example, a
single qubit system $U$ has an Euler angle decomposition of
\begin{equation}
U = R_z(\phi)R_x(\theta)R_z(\psi),
\end{equation}
up to an irrelevant global phase.  

Next, we evolve the system according to $e^{-iHt}$ for a controlled
time $t$. This results in the state
\begin{equation}
e^{-iHt} \ket{\Psi}=e^{-iHt}U \ket{0}.
\end{equation}

Finally, $U^\dagger$ is applied and a measurement in the
computational basis is performed, which returns '0' with probability
\begin{align}
P(0|H;t,U)&:=|\bra{0}U^\dagger e^{-iHt} U \ket{0}|^2\nonumber\\
 &= \left( \sum_k |\beta_k|^2 \cos{(\lambda_k t)} \right)^2 + \left( \sum_k |\beta_k|^2 \sin{(\lambda_k t)} \right)^2\nonumber\\
&= \sum_{ij} \cos{((\lambda_i -\lambda_j)t)}|\beta_i|^2|\beta_j|^2:=P(0|\Delta;t,U).
\end{align}
Note that the gaps $\Delta_{ij}= \lambda_i-\lambda_j$ cannot be easily
learned from this expression because of the unknown $\beta_i$ and
$\beta_j$ terms. Haar averaging provides a solution to this problem. 

Given an unknown Haar--random unitary $U$, the likelihood of
measuring '0' is given by the law of conditional probability
\begin{align}
P(0|\Delta;t)=\int P(0|\Delta;t,U) \mu(U) \mathrm{d}U,
\end{align}
where $\mu$ is the Haar--measure over $U(N)$. This probability is
distinct from the likelihood that would be used if the user knew, or
was capable of computing, $P(0|H;t,U)$ for the particular $U$ that
was used in the experiment.

If we define $\Delta$ to be a matrix such that
$\Delta_{i,j}:=\lambda_i-\lambda_j$, this Haar average evaluates to
\begin{align}
P(0|\Delta;t) :&=N\langle |\beta_i|^4 \rangle + \langle |\beta_i|^2
|\beta_j|^2 \rangle_{i\neq j} \sum_{i\neq j} \cos{(\Delta_{ij}t)} \nonumber\\
&=
\frac{2}{N+1}\left( 1 + \frac{1}{N}\sum_{i>j}\cos{(\Delta_{ij}t)}
\right).\label{eq:bageq}
\end{align}
Eq.~\eq{bageq} provides a likelihood function that can be used to
perform Bayesian inference. In particular, if a binary experiment is
performed wherein the only two outcomes are $\ket{0}$ and $\ket{v\ne
  0}$, the latter occurs with probability $1-P(0|H;t)$. If we define
our prior distribution over eigenvalues to be $P(\lambda)$, then given
a measurement of '0' is recorded, Bayes' rule states that the
posterior distribution is
\begin{equation}
P(\lambda|0;t) =\frac{2\left( 1 +
  \frac{1}{N}\sum_{i>j}\cos{(\Delta_{ij}t)}\right)P(\lambda)}{(N+1)P(0)},
\label{eq:likelihood}
\end{equation}
where $P(0)$ is a normalisation factor. We outline this approach in
\alg{inference0}. Thus Bayesian inference allows the gaps to be
learned from such experiments. 

The Cram\'er--Rao bound provides an estimate for the minimum number of
experiments and/or experimental time needed to ensure that the
variance of $\Delta_{ij}$ is sufficiently small~\cite{Cra45}.  If we
assume that $R$ experiments are performed, each with $D\in O(1)$
outcomes and evolution time at least $t$, then the elements of the
Fisher matrix are
\begin{equation}
I_{ij,kl} = \sum_{d_1=0}^{D-1}\cdots \sum_{d_R=0}^{D-1}
\frac{\partial_{\Delta_{ij}} \prod_{q=1}^R
  P(d_q|\Delta)\partial_{\Delta_{kl}} \prod_{q=1}^R
  P(d_q|\Delta)}{\prod_{q=1}^RP(d_q|\Delta)}\in O\left(\frac{R^2
  t^2}{N^2} \right),\label{eq:fisher}
\end{equation}
which through the use of the Cram\' er--Rao bound implies that the
variance of any unbiased estimator of $\Delta_{ij}$ after all $R$
experiments scales at least as $\Omega(N^2/R^2  t^2)$. 

While this shows that the number of experiments needed to learn the
gaps grows at least linearly with $N$, the uncertainty in the optimal
unbiased estimator also shrinks as $\Omega(1/T)$, where $T$ is the
total evolution time. The error scaling is hence quadratically smaller
than would be expected from statistical sampling. This opens up the
possibility of very high precision frequency estimates for small
quantum devices.

\subsection{Implementing Haar random unitaries}
\label{sec:haar_random}

In practice, Haar-random unitaries can be hard to exactly implement on a
quantum computer~\cite{ZK94}. This task is even more challenging in an
uncalibrated device. Fortunately, efficient methods for implementing
pseudo--random unitaries are well known.  Here we use the definition
of $t$--design given by Dankert et al~\cite{DCE09}, wherein a unitary
$t$--design on $N$ dimensions is taken to be a finite set of unitary
operations such that the average of any polynomial function of degree
at most $t$ in the matrix elements and their complex conjugates agrees
with the Haar average. In particular, let $P_{(t,t)}(U_k)$ be such a
polynomial function and assume that there are $K$ elements in this
design. Then
\begin{equation}
\frac{1}{K} \sum_{k=1}^K P_{(t,t)}(U_k) = \int P_{(t,t)}(U) \mu(U)
\mathrm{d}U,
\end{equation}
where $\mu$ is the Haar--measure.  

More generally, one can also consider the notion of an
$\epsilon$--approximate $t$--design.  This notion differs from a
$t$--design in that strict equality is not required and a discrepancy
of at most $\epsilon$ in the diamond--distance is permitted.
Specifically, let $G_W(\rho)=\sum_i U_i^{\otimes k} \rho
(U_i^\dagger)^{\otimes k}$ be a super-operator corresponding to
twirling the input operator over the elements of the design and
$G_H(\rho)=\int U^{\otimes k} \rho (U^\dagger)^{\otimes k}
\mu(U)\mathrm{d}U$ be the same quantity averaged over the
Haar--measure. Then the set of unitary operators forms a
$\epsilon$--approximate $k$--design if $\|G_W -G_H\|_\diamond \le
\epsilon$, where the diamond norm is discussed in detail
in~\cite{watrous2011theory}.

This definition implies that the Haar--expectations of the terms
in~\eq{bageq} can be found using a $2$-design since all such terms are
quadratic in the matrix elements of $U$ and their complex
conjugates. The remaining question is how to form such a design.

Dankert et al. show that random Clifford operations can be used to
form exact and $\epsilon$--approximate $2$-designs using $O(n^2)$
and $O(n\log(1/\epsilon))$ gates respectively~\cite{DCE09}. These results
are sufficient in settings where the device has access to well
characterised Clifford gates. If the system does not have such gates,
the result of Harrow and Low~\cite{HL09} can be
applied to show that $\epsilon$--approximate $t$--designs can be
formed out of random sequences of gates taken from any universal gate
set. The number of gates needed to generate such a pseudo--random
unitary is $O(n(n+\log(1/\epsilon)))$~\cite{HL09} in such cases, where
the multiplicative constant depends on the details of the universal
gate set used.

The result of Harrow and Low is especially significant here because
the gates do not need to be known for it to hold (however we still require that the exact inverses
of the uncalibrated gates can be performed). This means that even
if we want to use randomised gap estimation to learn how to control certain
uncalibrated quantum systems then we do not need to explicitly know the
gates that are actually applied to the system to sufficiently
randomise the measurement results if the underlying gate set is universal and a sufficiently long sequence of random operations is used.

\subsection{Numerical tests}
While the prior argument suggests that the error should scale as $1/T$
for randomised gap estimation, an important question remains regarding
how well it actually scales in practice. We assess this by using
rejection filtering to estimate the gaps and the particle guess
heuristic, discussed in~\cite{qp2,qp3,qp1} and also the appendix, to
set the time for each experiment. We use this approach because it is
fast, accurate and more importantly easier to implement than
sequential Monte--Carlo methods~\cite{qp4,stenberg2014efficient}.

Although performing rejection filtering to infer the gaps may seem
straight forward, a complication emerges that makes it conceptually
more challenging to apply the technique directly. To see this,
consider a Gaussian prior over the eigenvalue gaps for an
$N$--dimensional system. If the eigenvalue gaps are drawn
independently, the result will with high probability have inconsistent
eigenvalue gaps. By inconsistent we mean that if $\Delta_{43} =
\lambda_4 - \lambda_3$ and $\Delta_{32} = \lambda_3-\lambda_2$ are
eigenvalue gaps, then $\Delta_{42} = \lambda_4-\lambda_2 = \Delta_{43}
+ \Delta_{32}$ must also be one of the eigenvalue gaps for the system.
If $\Delta_{43}$ and $\Delta_{32}$ are chosen independently from a
Gaussian prior, then it is very unlikely that their sum is
also. Hence, gaps chose independently from the prior distribution will
not correspond to a feasible set of eigenvalues.

\begin{figure}[t!]
  \centering
  \includegraphics[width=0.32\columnwidth]{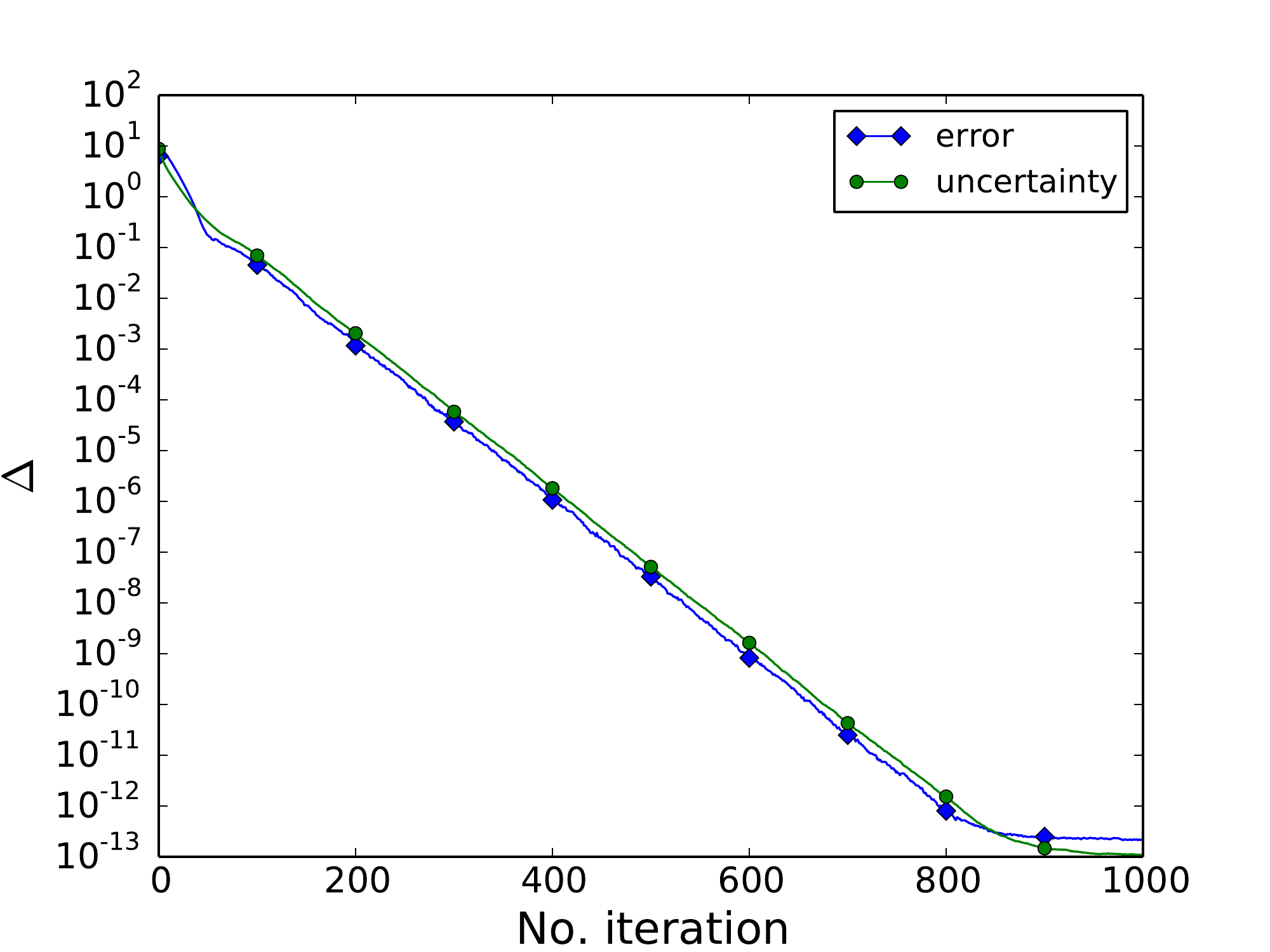}
  \includegraphics[width=0.32\columnwidth]{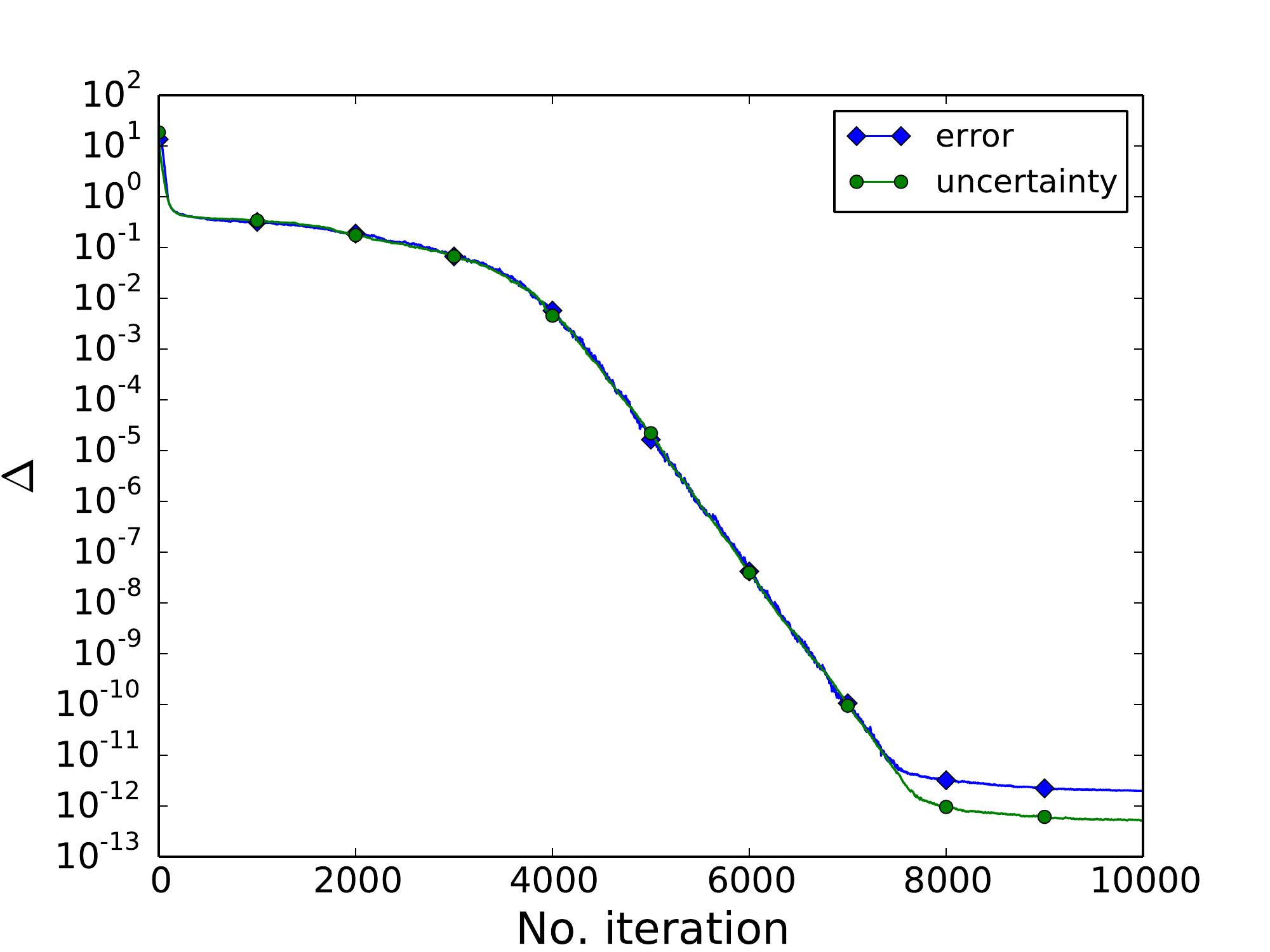}
  \includegraphics[width=0.32\columnwidth]{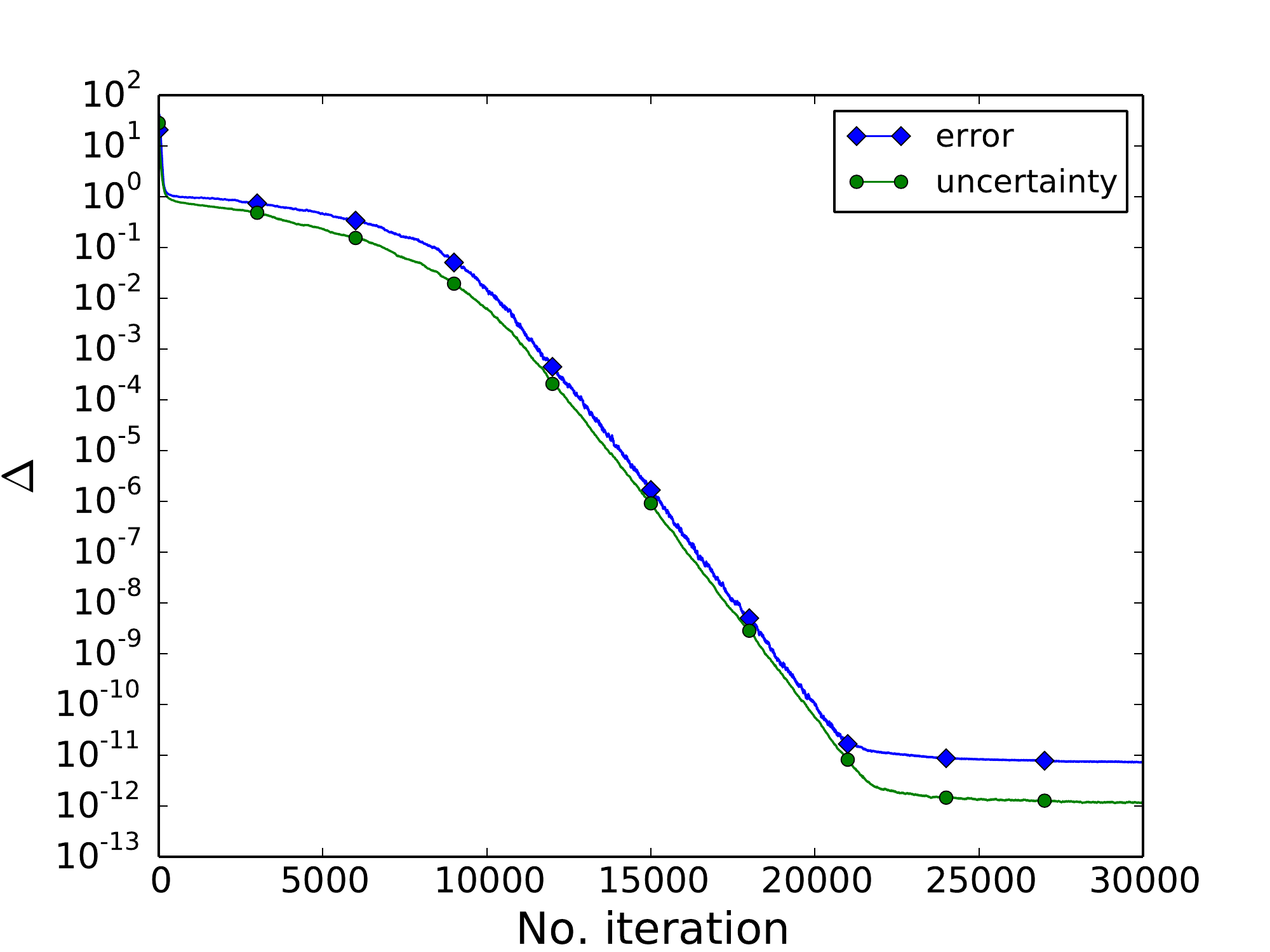}
  \caption{Median error and uncertainty in the eigenvalues computed
    using randomised gap estimation using rejection filtering for
    2-level (left), 3-level (middle) and 4-level (right) Hamiltonians
    with eigenspectra drawn from the Gaussian Unitary
    Ensemble.}
\label{fig:bag_of_gaps}
\end{figure}

There are many ways that this self-consistency constraint could be
imposed in the inference step, but perhaps the simplest approach is to
use Bayesian inference to learn a set of consistent eigenvalues for
the system, rather than gaps directly. The solution is not
unique~\cite{SSL90} as we discuss further in~\app{turnpike}. However,
any gap inferred from the potentially erroneous spectrum will be
automatically consistent. We enforce this by taking the
lowest--eigenvalue to be zero without loss of generality and then
define the error in the inference as
\begin{equation}
\Delta=\textnormal{min}\left\{\sum_i|\lambda_i -
\mu_i|,\sum_i|\lambda_i - (\textnormal{max}\{\mu_i\}_i-\mu_i)|\right\}
\label{eq:error}
\end{equation}
and the uncertainty as 
\begin{equation}
\sigma = \sqrt{\textnormal{Tr}(\Sigma)},
\label{eq:uncertainty}
\end{equation}
given that rejection filtering reports a covariance matrix $\Sigma$
with a mean of $\mu_i$ for each eigenvalue $i$. 
The minimisation of
the error in~\eq{error} is meant to remove the mirror symmetry that
naturally arises as any sorted eigenspectrum gives the same
distribution of gaps as its reflection about the midpoint
(see~\app{turnpike}). We find numerically that Bayesian inference can
rapidly estimate the eigenvalues of an unknown Hamiltonian, up to a
degeneracy, in one and two qubit systems (see~\fig{bag_of_gaps}). In
particular, we see that the error shrinks exponentially with the
number of measurements before being ultimately limited by
machine--$\epsilon$. This illustrates that randomised gap estimation
can learn the eigenvalue gaps for low--dimensional systems using
exponentially less data than would ordinarily be required by
statistical sampling, although with only a quadratic advantage in the
total experimental time.

For higher--dimensional systems achieving the same task is expected to
be much more daunting and the error metric in~\eq{error} should be the
minimum over all consistent re--orderings of the spectrum. Such
problems can, however, be sidestepped in systems with a gapped
adiabatic path connecting the Hamiltonian to a efficiently diagonalizable Hamiltonian, as we show
below.

\section{Adiabatic elimination of eigenvalues}
\label{sec:adiabatic}
\begin{algorithm}[t!]
  \caption{Bayesian inference algorithm for eigenphases}
  \label{alg:inference1}
  \begin{algorithmic}
    \Require Set of unitary operators $\{U: j=1:K\}$ such that
    $U=\openone\otimes V_j$ where $V$ acts on a $M\le N$ dimensional
    subspace of the computational basis, a diagonal Hamiltonian $H_0$,
    annealing time $T$, an interpolation function $f:[0,1]\mapsto [0,1]$
    such that $f(0)=1-f(1)=0$, evolution time $t$ and prior over $M$
    smallest gaps gaps $P(\Delta)$.

    \State Prepare state $\ket{0}\in \mathbb{C}^N$.

    \State Randomly select $U$ from set of unitary operations.

    \State $\ket{0}\gets U^\dagger \left(\mathcal{T} e^{-i \int_0^1
      (1-f(s))H_0 +f(s)H \mathrm{d}s T}\right)^\dagger e^{-iHt}
    \left(\mathcal{T} e^{-i \int_0^1 (1-f(s))H_0 +f(s)H \mathrm{d}s
      T}\right)U \ket{0}$.

    \State Measure state and $E\gets 0$ if the result is '0' otherwise
    $E\gets 1$.

    \State Use approximate Bayesian inference to estimate
    $P(\Delta|E)\propto P(\Delta) \cdot \int P(E|\Delta;t,U) \mu(U)\mathrm{d}U$.

    \State\Return $P(\Delta|E)$.
  \end{algorithmic}
\end{algorithm}

Most eigenvalue estimation tasks in quantum computing focus on
learning only a part of the spectrum rather than all the
eigenvalues. We could use randomised gap estimation to learn the
eigenvalues in this part of the spectrum if we could apply
Haar--random unitary operations in \emph{only} that subspace. Since
the eigenvectors of the Hamiltonian are generally not known, it is
difficult to do this directly. Despite this fact, the adiabatic
theorem provides a means by which these random unitaries can be
applied to the appropriate subspace by employing adiabatic state
preparation.

Let us assume we wish to learn the spectrum, up to an additive
constant, for a given subspace of a Hamiltonian $H_p$. In particular
let $S = {\rm span}(\ket{\lambda_{j_1}},\ldots,\ket{\lambda_{j_m}})$
be a subspace of eigenvectors of $H_p$ such that the eigenvalues obey
$\lambda_1\le \lambda_2 \le \cdots$ and $j$ is a monotonically
increasing sequence on $\{1,\ldots, 2^n\}$. Then we define an
adiabatic interpolation to be a time--dependent Hamiltonian of the
form $H(s)$ such that $H(0)=H_0$ and $H(1)=H_p$ where $H(s)$ is at
least three-times differentiable for all $s\in (0,1)$ (where $s$ is
the dimensionless evolution time) and furthermore has a gapped
spectrum for all times in this interval. A common example of such a
Hamiltonian is
\begin{equation}
H(s)=(1-s)H_0 + sH_p.
\end{equation}
The adiabatic theorem~\cite{Kat50,CHW11,EH12} shows that if
$\ket{\lambda_{j}^0}$ is an eigenvector of $H_0$ whose eigenvalue
corresponds to that of $\ket{\lambda_j}$ in the sorted list of
eigenvalues,
\begin{equation}
\mathcal{T}e^{-i\int_0^1 H(s) \mathrm{d}s T}\ket{\lambda_{j_k}^0} =
e^{i\omega_{j_k}T}\ket{\lambda_{j_k}}+O(1/T),\label{eq:adiabatic}
\end{equation}
where $\mathcal{T}$ is the time-ordering operator. Hence, if we can
perform a Haar--random unitary on the subspace of eigenvectors of
$H_0$, denoted by $S_0={\rm
  span}(\ket{\lambda_{j_1}^0},\ldots,\ket{\lambda_{j_m}^0})$, we can
adiabatically transform the resultant state to a Haar--random state in
$S$, up to phases on each of the eigenvectors and error $O(1/T)$.

Now let $U$ be a Haar--random unitary acting on $S_0$. Then the new
adiabatic protocol for estimating the eigenphase is
\begin{equation}
P(0|H;t,U)=\left|\bra{0}U^\dagger \left(\mathcal{T}e^{i\int_1^0
  H(s) \mathrm{d}s T}\right)e^{-iH_pt}\left(\mathcal{T}e^{-i\int_0^1
  H(s) \mathrm{d}s T}\right)U\ket{0}\right|^2.\label{eq:adprob}
\end{equation}
Given $U \ket{0} = \sum_{j:\ket{\lambda_j}\in S_0} \alpha_j
\ket{\lambda_j^0}$ and $\tilde U \ket{0} :=
\sum_{j:\ket{\lambda_j}\in S} \alpha_j
\ket{\lambda_j}$,~\eq{adiabatic} shows that
\begin{equation}
\left(\mathcal{T}e^{-i\int_0^1 H(s) \mathrm{d}s T}\right)U \ket{0} =
\sum_{j:\ket{\lambda_j}\in S_0} \alpha_je^{i\omega_j T}
\ket{\lambda_j}+O(1/T).  \label{eq:adsubs}
\end{equation}
Then using the fact that $H_p \ket{\lambda_j}
=\lambda_j\ket{\lambda_j}$, we see from~\eq{adprob} and~\eq{adsubs}
that
\begin{align}
P(0|H;t,U)&= \left( \sum_{k:\ket{\lambda_k} \in S} |\alpha_k|^2
\cos{(\lambda_k t)} \right)^2 + \left( \sum_{k:\ket{\lambda_k} \in S}
|\alpha_k|^2 \sin{(\lambda_k t)} \right)^2+O(1/T).\nonumber\\ &=
\left|\bra{0}\tilde U^\dagger e^{-iH_pt}\tilde
U\ket{0}\right|^2+O(1/T).
\end{align}
The adiabatic theorem can therefore be used to allow randomised gap
estimation to be performed on a specified subspace of eigenvalues if
there exists a gapped adiabatic path between a Hamiltonian that is
diagonal in the computational basis and the problem Hamiltonian
(see~\alg{inference1} for an outline of this approach). This shows
that the curse of dimensionality that afflicts this method can be
excised in cases where such an adiabatic evolution is
possible. Randomized gap estimation can thus, in principle, be used as
a fundamental tool to characterise and calibrate untrusted quantum
systems. We will examine this further in subsequent sections.

As a final note, the asymptotic scaling of the error that emerges
because of diabatic leakage out of the eigenstates, or more generally
eigenspaces, can be exponentially improved using boundary cancellation
methods~\cite{LRH09,RPL10,WB12,KW14}, which set one or more
derivatives of the Hamiltonian to zero at the beginning and end of the
evolution. However, while this can substantially reduce the cost of
the adiabatic transport, it requires fine control of the system
Hamiltonian in order to realise the benefits of the
cancellation~\cite{WB12}. As a result, it is not necessarily clear when
the higher-order versions of these algorithms will find use in
applications outside of fault tolerant quantum computing.

\section{Application to amplitude estimation}
\label{sec:amplitude}
Randomised gap estimation also provides an important
simplification for amplitude estimation, which is a quantum algorithm
in which the probability of an outcome occurring is estimated by
combining ideas from amplitude amplification and phase
estimation~\cite{BHMT02}. The algorithm quadratically reduces the number of queries
needed to learn a given probability.  The main significance of this is that it quadratically speeds up Monte--Carlo algorithms.  Specifically,
imagine that you are given an unknown quantum state of the form
\begin{equation}
  \ket{\psi}:=A\ket{0} = a\ket{\phi} + \sqrt{1-|a|^2}\ket{\phi^\perp},
  \label{eq:Aeq}
\end{equation}
where $\ket{\phi}\in \mathbb{C}^N$, $A$ is a unitary operator and $|a|
\in (0,1)$. Furthermore, assume that you are given access to an oracle
such that $\chi\ket{\phi}=-\ket{\phi}$ and for all states $\ket{v}$
orthogonal to $\ket{\phi}$, $\chi\ket{v} = \ket{v}$. Amplitude
estimation then allows $|a|^2$ to be estimated to within error
$\epsilon$ using $\tilde O(1/\epsilon)$ applications of $\chi$ and
$O(\log(1/\epsilon))$ measurements.

In order to understand how this works, consider the Grover search
operator $Q=-A\chi_0A^\dagger \chi$, where $\chi_0$ acts as $\chi$
would for $\phi=0$. For the case of a single marked state we then have
that
\begin{equation}
Q^j \ket{\psi} = \sin([2j+1]\theta_a)\ket{\phi} +
\cos([2j+1]\theta_a)\ket{\phi^\perp},
\end{equation}
where $\theta_a := \sin^{-1}({a})$.  It is then clear that $Q$ enacts
a rotation in this two--dimensional subspace and has eigenvectors
\begin{equation}
\ket{\psi_\pm} = \frac{1}{\sqrt{2}} \left(\ket{\phi} \pm i
\ket{\phi^\perp} \right),
\end{equation}
where
\begin{equation}
Q\ket{\psi_\pm} = e^{\pm i 2\theta_a}\ket{\psi_\pm}.
\end{equation}
All other eigenvectors in the space orthogonal to this have eigenvalue
$\pm 1$.

We can learn these phases using randomised gap estimation.
The most significant change is that here $t$ must be taken to be an
integer since fractional applications of the Grover oracle have not
been assumed. In this case, the eigenvalues are in the set
$\{-2\theta_a,0,2\theta_a, \pi\}$, which implies that $|\Delta_{i,j}|$
takes the values $\{0,\pm 2\theta_a,4\theta_a, \pi \pm 2\theta_a,\pi\}$.
This means that the gap estimation process has very different
likelihoods if $t$ is even or odd.  Since $\cos([\pi \pm
  2\theta_a](2p+1)) = -\cos(2\theta_a(2p+1))$, for integer $p$, many
of the terms in~\eq{bageq} cancel if $t$ is odd.  As a result, it is
better to choose $t$ to be an even integer for this application, in
which case the likelihood function is
\begin{equation}
P(0|\theta_a;t)=\frac{2}{N+1} \left(1 +
\frac{1}{N}\left(\binom{N-2}{2}+2(N-2)\cos(2\theta_a t)+
\cos(4\theta_a t) \right) \right).
\end{equation}
Such experiments cannot yield substantial information as $N\rightarrow
\infty$ according to~\eq{fisher} and thus different ideas are needed
to make large $N$ instances tractable unless we can sample according
to the Haar measure only within the relevant subspace using techniques
similar to those in~\sec{adiabatic}.

This method is, however, viable without modification for small $N$. To
illustrate this, consider the case where $N=2$, where the likelihood
function reduces to
\begin{equation}
P(0|\theta_a; t) = \frac{2}{3}\left(\frac{1}{2} + \cos^2(2\theta_a t)
\right).\label{eq:PE2}
\end{equation}
This likelihood function yields only slightly less information than the one for
iterative phase estimation, $P(0|\theta_a;t) = \cos^2(2\theta_a t)$,
as their derivatives with respect to $\theta_a$ differ by at most a
constant factor. This observation, coupled with the fact that an
additional qubit is not needed, means that amplitude estimation can be
realistically applied in-place in single qubit devices using this
method. 

\section{Application to control map learning}
\label{sec:map}
\begin{figure}[t!]
    \includegraphics[width=0.98\columnwidth]{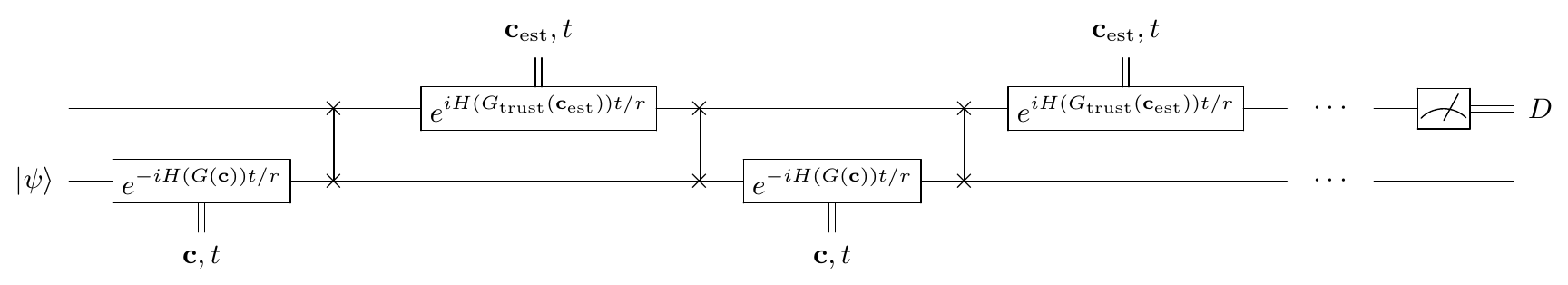}
    \caption{Quantum bootstrapping experiment to learn a control map
      using interactions between a trusted simulator and an
      uncharacterised device.  The value $r$ can be increased to
      combat the spread of non--local correlations and also to reduce
      ambiguities that can arise in the inference process when the
      untrusted and trusted Hamiltonians do not commute.}
    \label{fig:QBS}
\end{figure}
Another important application of randomised gap estimation is
calibration of quantum devices. In order to understand why calibration
can be a problem, imagine that you wish to calibrate a qubit so that
you can perform an $X$ gate within error $\epsilon = 10^{-5}$. In
order to verify, using with plain statistical sampling, that the error
is this small, roughly $10^{10}$ experiments are needed. This process
would also have to be repeated for every quantum bit in the system,
which could easily result in terabytes of data needed to calibrate a
small fault tolerant quantum device. Our approach provides a method
that can exponentially reduce the number of measurements needed to
control such a qubit under these circumstances and polynomially reduce
the experimental time needed.

In order to abstract the problem of calibrating an unknown quantum
device, let us consider the problem of learning a control map. We
define the control map to be the mapping between a set of experimental
controls (specified by the vector $\mathbf{c}$) and the system
Hamiltonian. Specifically, if we define $H(\mathbf{x})=\sum_{i=1}^N
x_i H_i$ for a set of Hamiltonians $H_i$, then $\mathbf{x}$ is given
by $\mathbf{c}$ via
\begin{equation}
\mathbf{x} = G\mathbf{c} + \mathbf{x}_0,
\end{equation}
where the matrix $G$ is known as the control map. For simplicity we
will only consider the linear case where $\mathbf{x}_0=0$ below. The
general case is handled by first learning $\mathbf{x}_0$ and is
discussed in~\cite{qp1}.

Recently, a method called quantum bootstrapping has been proposed to
learn such maps efficiently, wherein a trusted quantum simulator is
used to infer the controls of an untrusted device (see~\fig{QBS} for a
circuit diagram of a bootstrapping experiment). Although this can
allow the control map to be learned using a poly-logarithmic number of
experiments, it necessitates the use of swap gates that couple the
device to the trusted simulator that is used to characterise it. This
means that verifying that the trusted simulator is properly working
can still require an exponential number of measurements (even under
assumptions of locality). This is significant because the error in the
trusted simulator places a lower limit on the error that can be
attained using the experiments proposed in~\cite{qp1,qp2,qp3,qp4}.
\begin{figure}[t!]
  \centering
  \includegraphics[width=0.55\columnwidth]{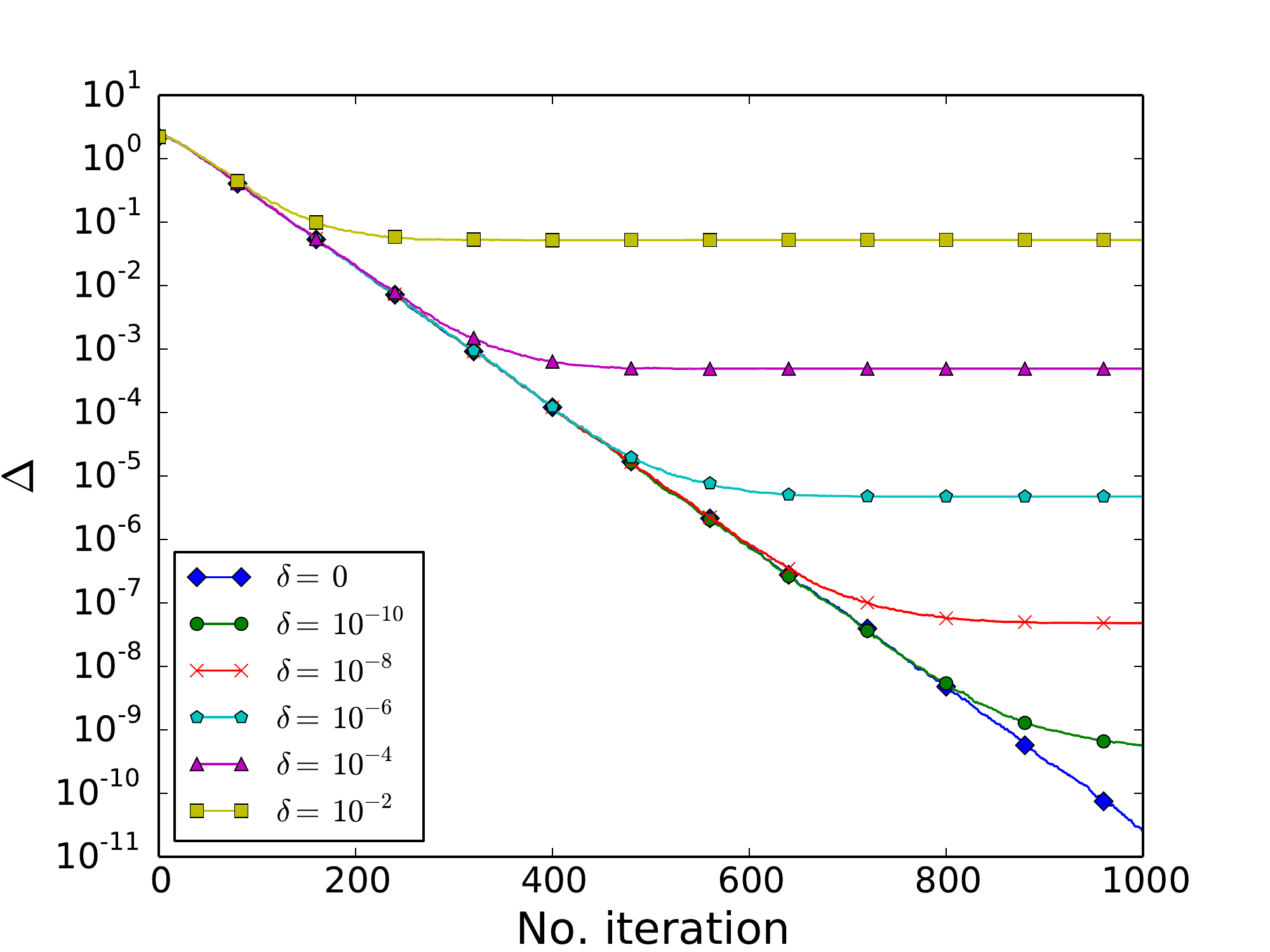}
  \caption{Median error for a quantum bootstrapping protocol that attempts to learn a control map for a universal qubit
    using a miscalibrated second qubit to control the measurement
    basis with miscalibrations of size $\delta$. These results are simulated in the limit as $r\rightarrow
    \infty$ (see~\fig{QBS}) as rejection sampling filtering has poor
    success probability for non--commuting models because of (intermediate multi-modality), unlike the
    sequential Monte-Carlo methods with $r=1$ studied in~\cite{qp3}.}
\label{fig:map}
\end{figure}
\fig{map} demonstrates the dependence of the error in the inference of
a control map for a single qubit that arises from the use of a
miscalibrated trusted simulator that has $G$ drawn from a zero--mean
Gaussian with covariance matrix $\delta \openone$, where $\delta$ is
chosen to take values within the interval $[0,10^{-2}]$. We see there
    that the error decays until it saturates at a level dictated by
    the miscalibrations in the trusted Hamiltonian.  This illustrates
    that errors in the trusted simulator can prevent a bootstrapping
    protocol from perfectly controlling this idealised quantum system.

It is worth noting that although Bayesian inference does not correctly
infer the true control map for the system because of these errors in
the trusted simulator, it does infer a set of controls that precisely
mimics the trusted device. In this sense, the learning task can be
thought of as a machine learning task, wherein the untrusted device is
trained to maximise the fidelity of its output with that of the
``trusted'' simulator, wherein the output state can be thought of as
the training data for the problem. It is therefore crucial to have a
well-calibrated trusted device if we wish to leverage the full power
of quantum bootstrapping.

To this end we provide a number of examples below that show how the
control maps of single qubit Hamiltonians can be learned using
randomised gap estimation.  We thereby show that our techniques may be
of great value in building a trusted simulator that can be used in
bootstrapping protocols.

\subsection{Two level example}
In the case for a two--level, single--qubit, system, the Hamiltonian
can be written, up to an irrelevant shift in energy, as
\begin{equation}
H([\alpha,\beta,\gamma])=\alpha X + \beta Y + \gamma Z.
\end{equation}
Since these three operations anti--commute it is evident that the
eigenvalues of $H$ are
\begin{equation}
E=\pm \sqrt{\alpha^2 +\beta^2 +\gamma^2}.
\end{equation}
Thus we can infer information about $\alpha$, $\beta$ and $\gamma$,
and in turn $G$, from the eigenspectra of different experiments.
Specifically,
\begin{equation}
\begin{bmatrix}
\alpha\\
\beta\\
\gamma
\end{bmatrix}
=
\begin{bmatrix}
G_{00} & G_{01} & G_{02}\\
G_{10} & G_{11} & G_{12}\\
G_{20} & G_{21} & G_{22}
\end{bmatrix}
\begin{bmatrix}
c_1\\
c_2\\
c_3
\end{bmatrix}.
\end{equation}

The simplest example of control map learning is the diagonal case, for
which
\begin{equation}
E= \pm \sqrt{G_{00}^2 c_1^2 +G_{11}^2 c_2^2 + G_{22}^2 c_3^2 }.
\end{equation}
and the control map can be learned, up to signs of $G_{ij}$, using a
sequence of three randomised phase estimation experiments with
$c=[1,0,0], [0,1,0]$ and $[0,0,1]$. The signs can be inexpensively
learned using auxiliary experiments because only a single bit of
information is required. Thus empirically the number of experiments
required to learn $G$ using randomised phase estimation is upper
bounded by a constant times $\log(1/\epsilon)$ (see
~\fig{bag_of_gaps}).

The next simplest case to consider is that of an upper triangular
control map with a positive diagonal in the form
\begin{equation}
\begin{bmatrix}
\alpha\\
\beta\\
\gamma
\end{bmatrix}
=
\begin{bmatrix}
G_{00} & G_{01} & G_{02}\\
0 & G_{11} & G_{12}\\
0 & 0 & G_{22}
\end{bmatrix}
\begin{bmatrix}
c_1\\
c_2\\
c_3
\end{bmatrix}.
\end{equation}
Clearly $G_{00}$ can be learned directly with randomised gap
estimation using the prior experiments. However, the remaining
elements must be inferred from different different experiments. If
$c=[1,1,0]$ then
\begin{align}
E([1,1,0])&=\pm \sqrt{(G_{00}+G_{01})^2+G_{11}^2}\\
E([0,1,0])&=\pm \sqrt{G_{01}^2+G_{11}^2}.\label{eq:010}
\end{align}
After squaring and subtracting both equations we find
\begin{equation}
E^2([1,1,0])-E^2([0,1,0]) = G_{00}^2 + 2G_{00}G_{01},
\end{equation}
which can be solved uniquely for $G_{01}$ since $G_{00}$ is known.
Once $G_{01}$ is known, $G_{11}$ can be learned unambiguously
from~\eq{010} using the fact that $G_{11}\ge 0$.

After these steps we have inferred the first two columns of $G$. The
remaining column can be learned similarly by performing three
randomised gap estimation experiments with $c=[0,0,1],[1,0,1]$ and
$[0,1,1]$ which yield
\begin{align}
E([0,1,1])&=\pm \sqrt{(G_{01}+G_{02})^2+(G_{11}+G_{12})^2 +
  G_{22}^2}\label{eq:011}\\
E([1,0,1])&=\pm \sqrt{(G_{00}+G_{02})^2+G_{12}^2 +
  G_{22}^2}\label{eq:101}\\
E([0,0,1])&=\pm \sqrt{G_{02}^2+G_{12}^2 + G_{22}^2}\label{eq:001}
\end{align}
Then by subtracting the square of~\eq{001} from~\eq{101} we learn
$G_{02}$, from which we learn $G_{12}$ from substituting the result
into~\eq{011}.  $G_{22}$ can then be learned by subtracting the square
of~\eq{101} from~\eq{011}, substituting the value of $G_{02}$ and
using $G_{22}\ge 0$.

More generally, this approach provides information about the inner
product between any two columns of $G$. No more information about $G$
can be extracted with such measurements. That is, $G^TG$ can be
learned, but $G$ itself is only determined up to an orthogonal
transformation $G^\prime = QG$ which preserves the inner products of
its columns. The matrix $Q$ can further be determined only if
additional constraints are imposed on $G$. For example, if $G$ is
upper or lower triangular and with a positive diagonal it is unique
within such measurements and $Q$ is the identity matrix as discussed
above.

\subsection{Learning $2\times 2$ single qubit control maps}
As mentioned above, gap estimation experiments alone are in general
not enough to uniquely specify the control map. In fact, amplitudes of
the states in the eigenbasis are required. This raises the question of
whether amplitude estimation may be used to glean the necessary
information from these coefficients using $O(\log(1/\epsilon))$
measurements.

For simplicity lets consider the Hamiltonian of a single qubit with
only $X$ and $Z$ rotations
\begin{equation}
H([\alpha,\gamma])=\alpha X + \gamma Z.
\end{equation}
and
\begin{equation}
\begin{bmatrix}
\alpha\\
\gamma
\end{bmatrix}
=
\begin{bmatrix}
G_{00} & G_{01}\\
G_{10} & G_{11}\\
\end{bmatrix}
\begin{bmatrix}
c_1\\
c_2
\end{bmatrix}.
\end{equation}
The case of a full $3\times 3$ control map is similar, but we do not
discuss it here because it is much more cumbersome owing to the
increased number of equations and unknowns.

Randomized gap estimation can provide
\begin{align}
E([1,0]) &= \pm \sqrt{G_{00}^2 + G_{10}^2}\nonumber\\
E([0,1]) &= \pm \sqrt{G_{01}^2 + G_{11}^2}\nonumber\\
E([1,1]) &= \pm \sqrt{(G_{00 }+G_{01})^2 + (G_{10}+G_{11})^2},
\label{eq:Eeqns}
\end{align}
Since there are four unknowns and three equations, the control map
cannot be unambiguously learned from these relations. However, it can
be seen that any other such randomised gap estimation provides
information that is not linearly independent of these three equation.
We need additional information that cannot directly come from
randomised gap estimation to solve the problem. This problem was
avoided in the prior example by constraining $G_{10}$ to be zero.

We can learn the required final piece of information by measuring
expectation values of the following three operators
\begin{align}
A([1,0])&=e^{-i(G_{00}X+G_{10}Z)\pi/(2E([1,0]))},\nonumber\\
A([0,1])&=e^{-i(G_{01}X+G_{11}Z)\pi/(2E([0,1]))},\nonumber\\
A([1,1])&=e^{-i((G_{00}+G_{01})X+(G_{10}+G_{11})Z)\pi/(2E([1,1]))}.
\end{align}
These each correspond to free evolution of the system Hamiltonian
under particular control settings for fixed durations.  We can then
use these operators (see eq~\eq{Aeq}) to learn the following
amplitudes while marking the $\ket{0}$ state
\begin{align}
  |a([1,0])|&=\left|\bra{0} A([1,0]) |\ket{0}\right| = \left|\frac{G_{10}}{E([1,0])}\right|,
  \\ |a([0,1])|&=\left|\bra{0} A([0,1]) \ket{0}\right| = \left|\frac{G_{11}}{E([0,1])}\right|,
  \\ |a([1,1])|&=\left|\bra{0} A([1,1])|\ket{0}\right| = \left|\frac{G_{10}+G_{11}}{E([1,1])}\right|.
\end{align}
If the sign of $G_{11}$ or $G_{10}$ is known, these three quantities
are enough to unambiguously solve for $G_{11}$ and $G_{10}$, and can
all be learned in-place using randomised amplitude estimation if an
accurate $Z$--gate can be performed.  The remaining elements of the
control map can then be found from~\eq{Eeqns}.  If the sign of
$G_{11}$ or $G_{10}$ is not known, they can still be learned with
amplitude estimation if the states $(\ket{0} \pm i\ket{1})/\sqrt{2}$
can be prepared on the quantum device.

In the absence of a calibrated $Z$ gate, the above strategy of amplitude
estimation does not apply. The requisite quantities $|a([0,1])|,
|a([1,0])|$ and $|a([1,1])|$ can nonetheless be learned by statistical
sampling. Doing so requires $O(1/\epsilon^2)$ experiments, which
dominates the cost of data acquisition. However, randomised gap
estimation efficiently provides $E([0,1]), E([1,0])$ and $E([1,1])$ which
allows us to optimally extract these expectation values (for the form
of the experiments considered). To see this consider the case of
learning $|a([0,1])|$ for arbitrary $t$. The likelihood function is
\begin{equation}
P(0|G;t)= \cos^2(E([0,1]) t)+ \sin^2(E([0,1])t)|a([0,1])|^2.
\end{equation}
It is then clear that the Fisher information is maximised at
$t=\pi/(2E([0,1]))$.  In order to compare, let us consider the case
where $E([0,1])t \mod 2\pi$ is taken from a uniform prior over
$0,2\pi$. Then the marginalised likelihood is
\begin{equation}
\frac{1}{2\pi} \int_0^{2\pi} \mathrm{d}x \cos^2(x) +\sin^2(x)
|a([0,1])|^2 = \frac{|a([1,0])|^2+1}{2}.
\end{equation}
thus~\eq{fisher} shows that the Fisher information for such
experiments is reduced by a factor of $2$. This furthermore reduces
the inference problem to frequency estimation wherein the
non--adaptive strategy of estimating $|a([0,1])|^2$ using the sample
frequency is optimal. Thus, even though randomised gap estimation does
not lead to a quadratic reduction in the experimental time needed to
learn $G$ here, it dramatically simplifies the inference problem and
removes the need to use multiple experimental settings to learn the
amplitudes.

Above we only discussed the case of a single qubit Hamiltonian with
uncalibrated $X$ and $Z$ interactions. The case of a general
(traceless) single qubit Hamiltonian follows from exactly the same
reasoning. First, one applies randomised gap estimation to learn the
relevant energies corresponding to the $6$ linearly independent sets
of controls. The remaining $3$ equations are found by measuring
appropriate expectation values using amplitude estimation if a $Z$
gate is available, and statistical sampling if it is not.

Finally, higher dimensional control maps can also be learned using
this approach. The analytic approach given above no longer holds
unless the Hamiltonian can be decomposed into a set of anti--commuting
Pauli operators. The same process can, however, be applied numerically
in such cases. We leave a detailed discussion of the issues that arise
when doing so for future work.

\section{Conclusion}
\label{sec:conclusion}
We introduced an alternative form of phase estimation that uses
Bayesian inference on top of randomised experiments to learn the
spectral gaps of a unitary operation without needing external qubits
or well-calibrated gates. In addition to calibration, this gap
estimation procedure allows amplitude estimation to be performed in
place for quantum systems without needing auxiliary qubits to control
the evolution.  These randomised gap and amplitude estimation
algorithms could radically reduce the costs of characterising and
calibrating small quantum devices. Our work further suggests that they
may find use in building the initial trusted simulator that is
required to start a quantum bootstrapping protocol~\cite{qp1} as well
as allowing amplitude estimation to be done in place in small quantum
systems.

An important caveat of our work is that the randomisation procedure that
we use to extract information about the gaps leads to an exponential decrease
in the signal yielded by experiments on high--dimensional systems.
We show that this can be mitigated, for some systems, using adiabatic
paths to allow the randomisation to be performed only within a low--dimensional eigenspace
of the system.  Regardless, this caveat still holds in general and as such
our methods are complimentary to existing phase and amplitude estimation
methods and do not replace them.  

This suggests an important remaining question namely that of whether efficient methods exist for estimating phase and amplitude
that do not require entangling the system with ancillary qubits.  The development of such methods would constitute a major step
forward for quantum metrology.  Similarly, proof that such methods are impossible would give us new insights into the
efficacy of phase estimation and reveal fundamental trade offs that exist between quantum resources for metrological tasks.

Finally, it should be noted that our work simply uses information from the Haar averaged probability distribution to infer the eigenvalue gaps.
In practice, additional information may be present in the higher order moments of the distribution.  Finding practical ways to exploit information carried by these higher moments about the phase remains an open problem.

\acknowledgments

We thank Matt Hastings and Shawn Cui for valuable comments and
introducing us to the turnpike problem as well as Michael Beverland, Alex Kubica, Christopher Granade and Paolo Zanardi for providing useful feedback on this work.

\appendix

\section{Numerics}
\label{app:numerics}
Here we provide detailed information about the numerical experiments
performed to generate~\fig{bag_of_gaps}. As mentioned in the main body
we use sets of eigenvalues as our hypotheses for Bayesian
inference. Without loss of generality we can shift the eigenspectrum
such that $\lambda_1 = 0$ and we hence have $N-1$ variables, where $N$
is the number of levels.

For the Bayesian inference step, we used rejection filtering with a
threshold of $10000$ accepted samples.  We found that drawing
parameters from the model until $10000$ samples are accepted provides
sufficient accuracy to represent the likelihood function in these
experiments.  We use a standard unimodal multivariate Gaussian as a
model for the prior distribution on hypothesis space. While this
unimodal distribution is incapable of representing a multi--modal
distribution (which may occur in cases where degenerate solutions
exist), multimodal models can also be used to improve
learning. However, for our experiments the mean and covariance of all
modes quickly converged to the same value and the accuracy was hence
no better than in the unimodal case.

The initial prior distribution is taken to be a Gaussian with mean and
standard deviation set to be those appropriate for Gaussian Unitary
Ensemble (GUE) Hamiltonians. In particular, we find these values using
the gap distribution of GUE Hamiltonians and set the mean eigenvalue
to be that found by adding $N-1$ random gaps drawn from the
distribution appropriate for GUE Hamiltonians to $\lambda_1=0$.
Similarly the covariance matrix for the $N-1$ unknown eigenvalues was
set to be variances equivalent to those found for each of the
eigenvalues of GUE Hamiltonians.  While this distribution does not
accurately model the true prior distribution of the eigenvalues, it
has sufficient overlap to be able to zero in on the correct
eigenvalues (up to degeneracies that cannot be resolved by randomised
gap estimation alone).

Evolution times are chosen to be $t = 1/(2\sigma)$, where $\sigma$ is
the uncertainty defined in~\eq{uncertainty}. This heuristic was
introduced in previous works~\cite{qp2,qp3,qp1} as the particle guess
heuristic and is shown to be nearly optimal for frequency estimation
problems. The prefactor of $1/2$ was found to work well empirically,
but more optimal choices may exist for this problem.

Since degeneracies in the eigenvalues can emerge in our problem that
prevent correct inference of the gaps, we resort to modifying the
likelihood function in order to forbid particular Hamiltonian models.
Specifically, we force the eigenvalues to be sorted and positive by
introducing an additional sign-term in the likelihood function
\begin{gather}
\langle L \rangle = \frac{2}{N+1}\left( 1 +
\frac{1}{N}\sum_{i>j}\textnormal{sign}(\Delta_{ij})\cos{(\Delta_{ij}t)}
\right)
\end{gather}
where the sign-function is $+1$ if its argument is positive and $-1$
otherwise and $\Delta_{ij}=\lambda_i - \lambda_j$. Hence, this
likelihood will align with the experimental likelihood at all
time-steps $t$ only if $\lambda_i > \lambda_j$. As this modified
likelihood can for some sets of parameters be larger than unity or
smaller than zero, we do an additional truncation step such that
$0 \leq L \leq 1$.

The median error and uncertainty reported are computed across 1000
instances drawn from the Gaussian Unitary Ensemble.  This number was
sufficient to make the sample error in the estimates graphically
insignificant.

\section{The turnpike problem for unique gaps}
\label{app:turnpike}
Eigenvalues are not directly provided in randomised phase
estimation. Instead the eigenvalues for the system must be inferred
from the gaps. In the main body, as well as the numerical experiments,
we use Bayesian inference to infer the eigenvalues from the
data. However, it remains an open question under what circumstances it
can successfully infer the eigenvalues, up to a constant shift, from
the gaps. The problem of inferring a set of numbers given only the
differences between them is well known in computer science and is
called the turnpike problem~\cite{SSL90}. The turnpike problem has
proven difficult to solve in general and the best algorithm to find a
solution remains an open problem. However, for the typical case,
wherein the gaps are unique, there are at most two degenerate choices
of spectra that could equivalently describe the system, as we show
below.
\begin{theorem}
Let $\{\lambda_1:=0 <\lambda_2 <\cdots<\lambda_N\}$ be a set such that
for any $i\ne k$ $(\lambda_{i+1}-\lambda_i )\ne
(\lambda_{k+1}-\lambda_k )$. Given $\Delta: \Delta_{xy} =
\lambda_{p_x}-\lambda_{q_y}$ for unknown permutations of
$\{1,\ldots,N\}$ $p$ and $q$, there are only two possible solutions
for the values of $\lambda$.
\end{theorem}
\begin{proof}
The nearest neighbour gaps can be isolated from this set easily.
Consider $\lambda_{i+2}-\lambda_i= \lambda_{i+2} - \lambda_{i+1} +
\lambda_{i+1} - \lambda_j$. Since $\lambda_{i+2}-\lambda_{i+1}$ and
$\lambda_{i+1}-\lambda_i$ are in $\Delta$ it follows that such a gap
can be decomposed into a sum of gaps in the set.  Similarly, by
induction $\lambda_{i+x} - \lambda_i$ can be decomposed into a sum of
nearest neighbour gaps for all $x\ge 2$.  In contrast, the nearest
neighbour difference $\lambda_{i+2}-\lambda_{i+1}$ cannot be decomposed
into a sum of other gaps because of our assumptions of uniqueness.
This means that $\Delta^1$, the set of all nearest neighbour gaps, can
be uniquely identified from this data.

It remains to show that there are only two ways that these gaps can be
laid out to form a spectrum that is  consistent with the data in
$\Delta$. The quantity $\Delta_{xy}$ $\lambda_N$ can be identified from the
spectrum since it is simply the largest gap under the assumption that
$\lambda_1=0$. The second largest gap can also be uniquely found from
this data. There are two possibilities. Either the second largest gap
is $\lambda_{N-1}$ or it is $\lambda_N - \lambda_2$.  This can be seen
by contradiction.  Imagine $\lambda_N - \lambda_j$ is the second
largest gap for $j \ge 3$.  Since $\lambda_j$ is a monotonically
increasing sequence, $\lambda_N -\lambda_j \le \lambda_N-\lambda_2 \le
\lambda_N$.  Therefore, it cannot be the second largest gap.  Applying
the same argument to $\lambda_{N-1} - \lambda_j$ for $j\ge 2$ leads to
the conclusion that $\lambda_{N-1}$ is the only other candidate for
the second largest gap.  Similarly $\lambda_{x} -\lambda_j \le
\lambda_{N-1} - \lambda_j$ for all $x\le N-1$ so these are the only
two possible candidates.

Assume that the second largest gap is $\lambda_{N-1}$.
Then $$\lambda_{N-2} = (\lambda_{N-2} - \lambda_{N-1}) +
\lambda_{N-1}.$$ Seeking a contradiction, imagine there is an
alternative solution to this such that
$$\lambda_{p}-\lambda_{N-1} = (\lambda_{r-1} - \lambda_{r}).$$ If
$\lambda_p -\lambda_{N-1}$ is not a nearest--neighbour gap, we have a
contradiction, since the set of nearest--neighbour gaps is known and
each gap is unique. Hence, we can immediately deduce that $\lambda_p$
is not a solution unless $p\in \{N-2,N\}$. Given $\lambda_p
-\lambda_{N-1}$ is a nearest neighbour gap we must have that $p=N-2$
or $p=N$. Since the nearest neighbour gaps are unique, $\lambda_{N-2}
-\lambda_{N-1} \ne \lambda_{N-1} - \lambda_N$ and since we already
know $\lambda_{N}-\lambda_{N-1}$ this confusion is impossible.
Furthermore, the uniqueness of the nearest neighbour gaps imply that
the only possible solution is $r=N-1$. Thus 
$\lambda_{N-2}$ is uniquely determined.

This process can be repeated using $\lambda_{p}$ and $\lambda_{p+1}$
for any $P<N$. Therefore, given the second largest gap is
$\lambda_{N-1}$, the spectrum can be uniquely determined under our assumptions.

Now let us assume the second largest gap is
$\lambda_{N}-\lambda_2$. Then $\lambda_2$ is known because
$\lambda_{N}$ is the largest gap, given $\lambda_1 = 0$.  By repeating the exact same argument as above we see
that $\lambda_{3}$ is uniquely determined by these two values and the
uniqueness of the nearest neighbour gaps.  This same argument can be
extended leading to the conclusion that if $\lambda_{2}$ is
$\lambda_N$ minus the second largest gap then the spectrum is also
unique. There are hence at most two possible spectra based on this
data.
\end{proof}
This approach shows that if the gaps are unique, a solution to the
turnpike problem is also unique up to a reflection of the spectrum
about its midpoint. However, in the presence of degenerate gaps the
number of possible solutions could scale slightly faster than linear
in $N$~\cite{SSL90} and the eigenspectrum can in this case only
narrowed down to be one of these solutions. Moreover, while a
backtracking algorithm can find a solution for the typical case in
time $O(N^2\log(N))$, the worst case complexity is $O(2^N N
\log(N))$~\cite{SSL90}. This suggests that Bayesian inference will not
always find a solution in time that is polynomial in $N$, which itself
is exponential in the number of qubits. However, for the two level
spectrum of a single qubit, the phase can be determined unambiguously.

\bibliographystyle{unsrt}
\bibliography{paper}

\begin{thebibliography}{10}

\bibitem{pe0}
Alexei~Yu Kitaev, Alexander Shen, and Mikhail~N Vyalyi.
\newblock {\em Classical and quantum computation}, volume~47.
\newblock American Mathematical Society Providence, 2002.

\bibitem{pe1}
K.~M. {Svore}, M.~B. {Hastings}, and M.~{Freedman}.
\newblock {Faster Phase Estimation}.
\newblock {\em ArXiv e-prints}, April 2013.

\bibitem{pe2}
Nathan Wiebe and Christopher~E Granade.
\newblock Efficient bayesian phase estimation.
\newblock {\em arXiv preprint arXiv:1508.00869}, 2015.

\bibitem{bonato2015optimized}
Cristian Bonato, Machiel~S Blok, Hossein~T Dinani, Dominic~W Berry, Matthew~L
  Markham, Daniel~J Twitchen, and Ronald Hanson.
\newblock Optimized quantum sensing with a single electron spin using real-time
  adaptive measurements.
\newblock {\em Nature nanotechnology}, 2015.

\bibitem{higgins2007entanglement}
Brendon~L Higgins, Dominic~W Berry, Stephen~D Bartlett, Howard~M Wiseman, and
  Geoff~J Pryde.
\newblock Entanglement-free heisenberg-limited phase estimation.
\newblock {\em Nature}, 450(7168):393--396, 2007.

\bibitem{hentschel2010machine}
Alexander Hentschel and Barry~C Sanders.
\newblock Machine learning for precise quantum measurement.
\newblock {\em Physical review letters}, 104(6):063603, 2010.

\bibitem{DCE09}
Christoph Dankert, Richard Cleve, Joseph Emerson, and Etera Livine.
\newblock Exact and approximate unitary 2-designs and their application to
  fidelity estimation.
\newblock {\em Physical Review A}, 80(1):012304, 2009.

\bibitem{kimmel2014robust}
Shelby Kimmel, Marcus~P da~Silva, Colm~A Ryan, Blake~R Johnson, and Thomas
  Ohki.
\newblock Robust extraction of tomographic information via randomized
  benchmarking.
\newblock {\em Physical Review X}, 4(1):011050, 2014.

\bibitem{kimmel2015robust}
Shelby Kimmel, Guang~Hao Low, and Theodore~J Yoder.
\newblock Robust calibration of a universal single-qubit gate set via robust
  phase estimation.
\newblock {\em Physical Review A}, 92(6):062315, 2015.

\bibitem{magesan2011scalable}
Easwar Magesan, Jay~M Gambetta, and Joseph Emerson.
\newblock Scalable and robust randomized benchmarking of quantum processes.
\newblock {\em Physical review letters}, 106(18):180504, 2011.

\bibitem{sheldon2016iterative}
Sarah Sheldon, Lev~S. Bishop, Easwar Magesan, Stefan Filipp, Jerry~M. Chow, and
  Jay~M. Gambetta.
\newblock Characterizing errors on qubit operations via iterative randomized
  benchmarking.
\newblock {\em Phys. Rev. A}, 93:012301, Jan 2016.

\bibitem{qp1}
Nathan Wiebe, Christopher Granade, and D~G Cory.
\newblock Quantum bootstrapping via compressed quantum hamiltonian learning.
\newblock {\em New Journal of Physics}, 17(2):022005, 2015.

\bibitem{BK98}
Xavier Boyen and Daphne Koller.
\newblock Tractable inference for complex stochastic processes.
\newblock In {\em Proceedings of the Fourteenth conference on Uncertainty in
  artificial intelligence}, pages 33--42. Morgan Kaufmann Publishers Inc.,
  1998.

\bibitem{DGA00}
Arnaud Doucet, Simon Godsill, and Christophe Andrieu.
\newblock On sequential monte carlo sampling methods for bayesian filtering.
\newblock {\em Statistics and computing}, 10(3):197--208, 2000.

\bibitem{LW01}
Jane Liu and Mike West.
\newblock Combined parameter and state estimation in simulation-based
  filtering.
\newblock In {\em Sequential Monte Carlo methods in practice}, pages 197--223.
  Springer, 2001.

\bibitem{WGKS15}
Nathan Wiebe, Christopher Granade, Ashish Kapoor, and Krysta~M Svore.
\newblock Bayesian inference via rejection filtering.
\newblock {\em arXiv preprint arXiv:1511.06458}, 2015.

\bibitem{Cra45}
Harald Cram{\'e}r.
\newblock {\em Mathematical methods of statistics}, volume~9.
\newblock Princeton university press, 1945.

\bibitem{ZK94}
Karol Zyczkowski and Marek Kus.
\newblock Random unitary matrices.
\newblock {\em Journal of Physics A: Mathematical and General}, 27(12):4235,
  1994.

\bibitem{watrous2011theory}
John Watrous.
\newblock Theory of quantum information.
\newblock {\em University of Waterloo Fall}, 2011.

\bibitem{HL09}
Aram~W Harrow and Richard~A Low.
\newblock Random quantum circuits are approximate 2-designs.
\newblock {\em Communications in Mathematical Physics}, 291(1):257--302, 2009.

\bibitem{qp2}
Nathan Wiebe, Christopher Granade, Christopher Ferrie, and D.~G. Cory.
\newblock Hamiltonian learning and certification using quantum resources.
\newblock {\em Phys. Rev. Lett.}, 112:190501, May 2014.

\bibitem{qp3}
Nathan Wiebe, Christopher Granade, Christopher Ferrie, and David Cory.
\newblock Quantum hamiltonian learning using imperfect quantum resources.
\newblock {\em Phys. Rev. A}, 89:042314, Apr 2014.

\bibitem{qp4}
Christopher~E Granade, Christopher Ferrie, Nathan Wiebe, and D~G Cory.
\newblock Robust online hamiltonian learning.
\newblock {\em New Journal of Physics}, 14(10):103013, 2012.

\bibitem{stenberg2014efficient}
Markku~PV Stenberg, Yuval~R Sanders, and Frank~K Wilhelm.
\newblock Efficient estimation of resonant coupling between quantum systems.
\newblock {\em Physical review letters}, 113(21):210404, 2014.

\bibitem{SSL90}
Steven~S Skiena, Warren~D Smith, and Paul Lemke.
\newblock Reconstructing sets from interpoint distances.
\newblock In {\em Proceedings of the sixth annual symposium on Computational
  geometry}, pages 332--339. ACM, 1990.

\bibitem{Kat50}
Tosio Kato.
\newblock On the adiabatic theorem of quantum mechanics.
\newblock {\em Journal of the Physical Society of Japan}, 5(6):435--439, 1950.

\bibitem{CHW11}
Donny Cheung, Peter H{\o}yer, and Nathan Wiebe.
\newblock Improved error bounds for the adiabatic approximation.
\newblock {\em Journal of Physics A: Mathematical and Theoretical},
  44(41):415302, 2011.

\bibitem{EH12}
Alexander Elgart and George~A Hagedorn.
\newblock A note on the switching adiabatic theorem.
\newblock {\em Journal of Mathematical Physics}, 53(10):102202, 2012.

\bibitem{LRH09}
Daniel~A Lidar, Ali~T Rezakhani, and Alioscia Hamma.
\newblock Adiabatic approximation with exponential accuracy for many-body
  systems and quantum computation.
\newblock {\em Journal of Mathematical Physics}, 50(10):102106, 2009.

\bibitem{RPL10}
AT~Rezakhani, AK~Pimachev, and DA~Lidar.
\newblock Accuracy versus run time in an adiabatic quantum search.
\newblock {\em Physical Review A}, 82(5):052305, 2010.

\bibitem{WB12}
Nathan Wiebe and Nathan~S Babcock.
\newblock Improved error-scaling for adiabatic quantum evolutions.
\newblock {\em New Journal of Physics}, 14(1):013024, 2012.

\bibitem{KW14}
M{\'a}ria Kieferov{\'a} and Nathan Wiebe.
\newblock On the power of coherently controlled quantum adiabatic evolutions.
\newblock {\em New Journal of Physics}, 16(12):123034, 2014.

\bibitem{BHMT02}
Gilles Brassard, Peter Hoyer, Michele Mosca, and Alain Tapp.
\newblock Quantum amplitude amplification and estimation.
\newblock {\em Contemporary Mathematics}, 305:53--74, 2002.

\end{thebibliography}
\end{document}